\documentclass[letterpaper, 10 pt, conference]{ieeeconf}
\IEEEoverridecommandlockouts
\usepackage{cite}
\usepackage[compact]{titlesec}
\usepackage{amsmath,amssymb,amsfonts}
\usepackage{hyperref}
\usepackage{graphicx}
\usepackage{textcomp}
\usepackage{times}
\usepackage{mathrsfs}
\usepackage{amstext} % Lots of math symbols and environments
\usepackage{algorithm}
\usepackage{algorithmic}
\usepackage{setspace} % For line spacing adjustments
\usepackage{amssymb} % For \qedsymbol
\usepackage{mathtools} 
\usepackage{verbatim}
\usepackage{booktabs} % Required for toprule, midrule, bottomrule
\usepackage[font=footnotesize]{caption}
\usepackage{subcaption}
\usepackage{color}
\renewcommand{\theenumi}{\alph{enumi}}

\newcommand{\real}{\mathbb{R}}

\newcommand{\Ac}{\mathcal{A}}

\newcommand{\Hc}{\mathcal{H}}

\newcommand{\Kc}{\mathcal{K}}

\newcommand{\Nc}{\mathcal{N}}

\newcommand{\Sc}{\mathcal{S}}

\newcommand{\Uc}{\mathcal{U}}
\newcommand{\Vc}{\mathcal{V}}

\newcommand{\Xc}{\mathcal{X}}

\newcommand{\rank}[1]{\operatorname{rank}(#1)}

\newcommand{\longthmtitle}[1]{\mbox{}{\textit{(#1):}}}

\newcommand{\oprocendsymbol}{\hbox{$\square$}}
\newcommand{\oprocend}{\relax\ifmmode\else\unskip\hfill\fi\oprocendsymbol}

\newtheorem{theorem}{Theorem}[section]

\newtheorem{assumption}[theorem]{Assumption}
\newtheorem{lemma}[theorem]{Lemma}

\newtheorem{remark}[theorem]{Remark}

\newtheorem{problem}[theorem]{Problem}

\renewcommand{\baselinestretch}{0.97}
\graphicspath{{fig/}}

\usepackage{multirow}  % Include this package for multi-row functionality



\newcommand{\marginJC}[1]{\marginpar{\color{red}\tiny\ttfamily JC: #1}}
\newcommand{\marginDS}[1]{\marginpar{\color{blue}\tiny\ttfamily DS: #1}}

\allowdisplaybreaks

\def\BibTeX{{\rm B\kern-.05em{\sc i\kern-.025em b}\kern-.08em
    T\kern-.1667em\lower.7ex\hbox{E}\kern-.125emX}}
\begin{document}
\title{\LARGE Koopman Subspace Pruning in Reproducing Kernel Hilbert
  Spaces via Principal Vectors\thanks{This work was supported by AFOSR
    Award FA9550-23-1-0740.}}

\author{Dhruv Shah \quad Jorge Cort\'es \thanks{D. Shah and
    J. Cort\'es are with Department of Mechanical and Aerospace
    Engineering, UC San Diego, USA, {\tt\small
      \{dhshah,cortes\}@ucsd.edu}}}

\maketitle
\thispagestyle{empty}

\begin{abstract}%
  Data-driven approximations of the infinite-dimensional Koopman
  operator rely on finite-dimensional projections, where the
  predictive accuracy of the resulting models hinges heavily on the
  invariance of the chosen subspace. Subspace pruning systematically
  discards geometrically misaligned directions to enhance this
  invariance proximity, which formally corresponds to the largest
  principal angle between the subspace and its image under the
  operator.  Yet, existing techniques are largely restricted to
  Euclidean settings. To bridge this gap, this paper presents an
  approach for computing principal angles and vectors to enable
  Koopman subspace pruning within a Reproducing Kernel Hilbert Space
  (RKHS) geometry. We first outline an exact computational routine,
  which is subsequently scaled for large datasets using randomized
  Nyström approximations. Based on these foundations, we introduce the
  Kernel-SPV and Approximate Kernel-SPV algorithms for targeted
  subspace refinement via principal vectors. Simulation results
  validate our approach.
\end{abstract}%
%\marginJC{Abstract is a bit dry, in that it gets directly to the paper
%  subject matter/contributions. We're missing 1-2 sentences to provide
%  context: e.g., why do we even care about computing principal angles?
%  Why does one prune? What role does that play in building
%  Koopman-based models of unknown dynamical systems?}
%

%% FOR FINAL VERSION
%\begin{IEEEkeywords}
%  Stability of nonlinear systems, neural feedback loops, Robust
%  control
%\end{IEEEkeywords}

\section{Introduction}
% Introduction and motivation

The Koopman operator \cite{AM-YS-IM:20} provides a powerful framework for analyzing complex nonlinear dynamical systems by representing the dynamics as an infinite-dimensional linear operator on function spaces. A
standard approach to approximating the Koopman
operator involves projecting it onto a finite-dimensional subspace
spanned by a specified dictionary of observables. Extended Dynamic
Mode Decomposition (EDMD) \cite{MOW-IGK-CWR:15} is a prominent
data-driven technique that relies on this strategy. However, the
fidelity of the resulting approximation (and the accuracy of the
extracted eigenfunctions) is highly sensitive to the choice of this
dictionary, often necessitating extensive trial-and-error and
parameter tuning. To alleviate this limitation, Kernel EDMD
\cite{MOW-CWR-IGK:14} was introduced to perform the projection with
respect to the inner product of a Reproducing Kernel Hilbert Space
(RKHS). Within this framework, the dictionary is implicitly defined by
kernel sections evaluated at the sampled data points. This inherently
data-driven parameterization naturally aligns the subspace with the
underlying dynamics, significantly enhancing the approximation quality
while reducing the hyperparameter tuning exclusively to the selection
of the kernel function. Supported by recently established error bounds
\cite{FK-FMP-MS-AS-KW:25}, Kernel EDMD has emerged as a highly
effective and popular paradigm for constructing Koopman-based
approximate dynamical models.

Despite its theoretical advantages, the primary limitation of exact
Kernel EDMD is its $\mathcal{O}(N^3)$ computational complexity, where
$N$ denotes the number of trajectory snapshots. This cubic scaling
renders the method computationally prohibitive for massive datasets
and large-scale dynamical systems. To circumvent this bottleneck,
recent literature has proposed scalable estimators utilizing
randomized kernel approximations
\cite{GM-AC-VK-PN-MP-LR:23,FN-SK:23}. These techniques have
demonstrated significant reductions in computational complexity while
maintaining strong approximation performance, thus enabling the
application of Kernel EDMD to larger datasets and more complex
systems.

A critical consideration when approximating the Koopman operator on a
finite-dimensional subspace is the degree to which the subspace is
invariant under the operator, a property formally termed invariance
proximity \cite{MH-JC:26-auto,MH-JC:24-csl-arxiv-revised}.  Because
practical implementations inherently rely on finite-dimensional
representations, minimizing this invariance proximity is essential for
accurate Koopman-based predictions. To address this, subspace pruning
methods \cite{MH-JC:25-access,DS-JC:26-cdc1} enhance the invariance
proximity of an initial dictionary by iteratively discarding
geometrically misaligned directions. This refinement substantially
improves the quality of the resulting Koopman-based approximation,
yielding more accurate eigenfunctions and superior predictive
performance. However, these pruning techniques have only been
developed for the standard EDMD framework with respect to the
empirical $L_2$ inner product. Extending subspace pruning to the
Kernel EDMD setting presents a compelling opportunity to unify the
strengths of both approaches: the rich, data-dependent dictionaries
inherent to RKHS and the enhanced approximation accuracy achieved
through targeted subspace refinement. Beyond improving invariance,
subspace pruning inherently serves as a powerful model reduction
technique for Kernel EDMD. Such dimensionality reduction is crucial
for synthesizing tractable controllers based on the lifted state
dynamics \cite{RS-KW-IM-JB-MS-FA:26}.

Finally, the subspace pruning approach proposed here is closely
related to an alternative line of research that deals with computing
Koopman residuals, Residual Dynamic Mode Decomposition
(ResDMD)~\cite{MJC-AT:24}. This framework has also been extended to
the RKHS setting by computing residuals with respect to the Koopman
dual \cite{MJK:24,NB-MJC-GC:25}. Building on these computational
tools, and independently from the present
manuscript,~\cite{GC-NB-JCL-SB-MJC:26} has recently introduced the
Principal Angle Decomposition (PAD) algorithm, that orders and retains
small-angle principal observables to construct reduced, more accurate
Koopman models.

\noindent \textbf{Statement of Contributions:} The overarching
contribution of this paper is the development of a rigorous, data-driven framework
for computing principal angles and vectors to enable Koopman subspace
pruning within a RKHS. Specifically, this contribution is structured
in three intertwined steps. First, we derive an exact RKHS
computational routine to compute the principal angles and vectors
between a chosen subspace and its Koopman image strictly with respect
to the RKHS inner product. Second, to overcome the computational
bottlenecks of the exact formulation, we propose a scalable Nystr\"om
approximation that yields a highly efficient and tractable algorithm
for large datasets.  Finally, we integrate these computational tools
into the Kernel-SPV and Approximate Kernel-SPV algorithms for targeted
subspace refinement via principal vectors. Simulation results
demonstrate empirical convergence of the Nystr\"om approximation to
the true principal arguments and showcase the effectiveness of our pruning
approach in improving quality of the approximated eigenfunctions.

%\marginJC{I think we should have a sentence where we say that these
%  results, combined/integrated with subspace prunning strategies,% give
%rise to Kernel-SPV and Approximate Kernel-SPV -- otherwise, it's %ugly
%to have those terms in the abstract but nowhere in the intro.}
%\marginDS{Resolved}

\section{Preliminaries}

We introduce\footnote{Let $\real^n$ denote the $n$-dimensional
  Euclidean space and $\mathcal{X} \subseteq \real^n$ the state
  space. For a matrix $M \in \real^{m \times n}$, we denote its range
  (column space) by $\mathcal{R}(M)$, its rank by $\rank{M}$, and its
  Moore-Penrose pseudoinverse by $M^\dagger$. For a Hilbert space
  $\Hc$ equipped with an inner product
  $\langle \cdot, \cdot \rangle_{\Hc}$ and induced norm
  $\|\cdot\|_{\Hc}$, subspaces of observables are denoted by
  calligraphic letters (e.g., $\Sc, \Vc$). The linear span and the
  dimension of a subspace are denoted by $\text{span}(\cdot)$ and
  $\text{dim}(\cdot)$, respectively. Throughout this paper, the terms
  ``function'' and ``vector'' are used interchangeably to refer to
  elements of the RKHS, reflecting their practical representation as
  linear combinations of empirical kernel sections. It will be clear
  from the context whether we are referring to a function in the RKHS
  or its coefficient vector representation. The precise algebraic
  formulation of operations within the RKHS (including inner products
  and linear combinations) is detailed in
  Section~\ref{subsec:RKHS_prelim}.}  here the concepts and
mathematical notation utilized throughout the paper. We begin by
reviewing the Koopman operator framework and its data-driven
approximation via orthogonal projection. Next, we define principal
angles and vectors to rigorously quantify the invariance proximity of
subspaces. Finally, we outline the geometry of the RKHS and formalize
the kernel EDMD procedure.

\subsection{The Koopman Operator and its Approximation}\label{sec:koopman-EDMD}
We consider a discrete-time dynamical
system on the state space $\mathcal{X} \subseteq \mathbb{R}^n$
described by a map $T: \mathcal{X} \to \mathcal{X}$:
\begin{equation}\label{eqn:sys_dynamics}
  x^+ = T(x), \quad x \in \mathcal{X}.
\end{equation}
The Koopman operator~\cite{AM-YS-IM:20}
$\mathcal{K}: \mathcal{F} \to \mathcal{F}$ is an infinite-dimensional
linear operator acting on a space of real-valued observables
$\mathcal{F} \ni \psi: \mathcal{X} \to \mathbb{R}$ as
$(\mathcal{K}\psi)(x) = \psi(T(x))$.  We assume that the function
space $\mathcal{F}$ is closed under composition with $T$. The space of
observables $\mathcal{F}$ is equipped with an inner product
$\langle \cdot, \cdot \rangle_{\mathcal{F}}$ and the associated norm
$\|\cdot\|_{\mathcal{F}}$. Let $\mathcal{S} \subset \mathcal{F}$ be a
subspace spanned by a finite set of linearly independent functions
$\Psi = \{\psi_1, \dots, \psi_s\}$ (the dictionary).

The goal of EDMD is to find a finite-dimensional operator
$K: \mathcal{S} \to \mathcal{S}$ that best represents the action of
$\mathcal{K}$ when restricted to $\mathcal{S}$. This is achieved by orthogonally projecting the image of the subspace back onto itself. Formally, let $P_{\mathcal{S}}: \mathcal{F} \to \mathcal{S}$ denote
the orthogonal projection operator onto $\mathcal{S}$. The EDMD
approximation is given by 
\begin{equation}\label{eqn:EDMD_approx}
K_{\text{EDMD}} \triangleq P_{\mathcal{S}} \mathcal{K}|_{\mathcal{S}}.  
\end{equation}
 
This projection interpretation of EDMD reveals a critical limitation:
if the subspace $\mathcal{S}$ is not invariant under $\mathcal{K}$
(i.e., $\mathcal{K}\mathcal{S} \not\subseteq \mathcal{S}$), the
projection $P_{\mathcal{S}}$ discards the component of the dynamics
that evolves orthogonal to $\mathcal{S}$, leading to error.

\subsection{Principal Angles and Vectors}\label{sec:principal-angles}
To rigorously quantify the alignment between the chosen subspace and
its evolution under the Koopman operator, we introduce the definitions
and properties of principal angles and vectors~\cite{AB-GHG:73}.

Given a Hilbert space \((\mathcal{H},\langle\cdot,\cdot\rangle)\), let
\(\mathcal{U},\mathcal{V}\subset\mathcal{H}\) be subspaces with
\(\dim(\mathcal{U})=d_1\) and \(\dim(\mathcal{V})=d_2\).  The
\emph{principal angles}
$0 \leq \theta_1 \leq \cdots \leq \theta_k \leq \frac{\pi}{2}$ between
$\mathcal{U}$ and $\mathcal{V}$, where $k = \min\{d_1, d_2\}$, are
defined recursively as follows:
\begin{align*}
  \cos \theta_j = \max_{u \in \mathcal{U},\, v \in \mathcal{V}} \;
  &
    |\langle
    u,
    v
    \rangle|
  \\ 
  \text{subject to} \quad & \|u\| = \|v\| = 1,
  \\
  & \langle u, u_i \rangle = 0, \; \langle v, v_i \rangle = 0, \,\,
    i = 1, \ldots, j-1, 
\end{align*}
where $u_i, v_i$ are the principal vectors corresponding to the
previous $(j-1)$ angles. The vectors $(u_j, v_j)$ achieving the
maximum are called the $j$-th pair of \emph{principal vectors}.

The principal vectors $\{u_j\}_{j=1}^k$ and $\{v_j\}_{j=1}^k$ are
orthonormal, i.e.,
\begin{equation*}  
  \langle u_i,u_j\rangle=\delta_{ij},\quad
  \langle v_i,v_j\rangle=\delta_{ij},\quad i,j=1,\dots,k.
\end{equation*}
Furthermore, the principal vectors can be extended to bases of their
subspaces. For instance, if $k=\dim(\Uc)$, then $\{u_j\}_{j=1}^k$ is
an orthonormal basis of $\Uc$, and $\{v_j\}_{j=1}^k$ can be extended
to an orthonormal basis of $\Vc$.  When working in the Euclidean
setting i.e., $\mathcal{H} = \real^n$, we can compute the principal
angles and vectors via the SVD of the
matrix formed by the inner products of the basis vectors of the two
subspaces, cf.~\cite[Theorem 1]{AB-GHG:73} for details.

\subsection{Invariance Proximity}\label{subsec:invariance_proximity}
To ensure accuracy of the Koopman approximation via EDMD, the chosen
finite-dimensional subspace $\mathcal{S}$ should be as close to
invariant as possible. The notion of invariance
proximity~\cite{MH-JC:26-auto,MH-JC:24-csl-arxiv-revised} quantifies how close a subspace is to
being invariant under the Koopman operator by measuring the largest
principal angle between the subspace and its image under the operator.
Formally, the \emph{invariance proximity} of a subspace $\mathcal{S}$
with respect to the operator $\mathcal{K}$ is defined by
\begin{equation}
  \delta(\mathcal{S}) \triangleq \sin \theta_{\max}(\mathcal{S},
  \mathcal{K}\mathcal{S}),
\end{equation}
where
$\mathcal{K}\mathcal{S} = \text{span}\{\mathcal{K}\phi \mid \phi \in
\mathcal{S}\}$ denotes the image of the subspace under the operator
and $\{\theta_i \}_{i = 1}^s \subset [0, \pi/2]$ are the principal
angles between $\mathcal{S}$ and $\mathcal{K}\mathcal{S}$.

Clearly, smaller invariance proximity $\delta(\mathcal{S})$ indicates
that the subspace $\mathcal{S}$ is closer to being invariant under
$\mathcal{K}$, which in turn leads to a more accurate EDMD
approximation. In fact, the invariance proximity directly corresponds
to the worst-case relative prediction error of the EDMD model over the
subspace. We can show that
\begin{equation}
  \delta(\mathcal{S}) = \sup_{\substack{f \in \mathcal{S}
      \\
      \|\mathcal{K}f\| \neq 0}} \frac{\| \mathcal{K}f -
    K_{\text{EDMD}}f \|}{\| \mathcal{K}f \|}.
\end{equation}

\subsection{Reproducing Kernel Hilbert Spaces}\label{subsec:RKHS_prelim}

Following~\cite{JHM-POA:15}, we consider a Reproducing Kernel Hilbert Space (RKHS) $\mathcal{N}(\mathcal{X})$ of functions $f: \mathcal{X} \to \mathbb{R}$ defined on a non-empty set $\mathcal{X} \subseteq \mathbb{R}^n$. This space is equipped with an inner product $\langle \cdot, \cdot \rangle_{\mathcal{N}(\mathcal{X})}$ and is uniquely associated with a symmetric, positive definite kernel $k: \mathcal{X} \times \mathcal{X} \to \mathbb{R}$. For any $x \in \mathcal{X}$, the kernel section $\Phi_x(\cdot) = k(\cdot, x)$ resides in $\mathcal{N}(\mathcal{X})$ and satisfies the reproducing property:
\begin{equation}
  f(x) = \langle f, \Phi_x \rangle_{\mathcal{N}(\mathcal{X})}, \quad \forall f \in \mathcal{N}(\mathcal{X}).
\end{equation}
Crucially, this structure allows us to evaluate inner products in the infinite-dimensional RKHS $\mathcal{N}(\mathcal{X})$ using simple kernel evaluations in the original state space, $\langle \Phi_{x_i}, \Phi_{x_j} \rangle_{\mathcal{N}(\mathcal{X})} = k(x_i, x_j)$. This is known as the kernel trick.

To streamline the algebraic manipulation of these infinite-dimensional
objects, we adopt a compact matrix-like notation for collections of
functions in $\Nc(\Xc)$. Let $F = [f_1 \, f_2 \, \dots \, f_p]$ and
$G = [g_1 \, g_2 \, \dots \, g_q]$ be collections of functions in
$\Nc(\Xc)$, represented as formal row vectors. We define their
cross-inner product matrix $F^\top G \in \real^{p \times q}$ as the
real-valued matrix whose $(i, j)$-th entry evaluates the RKHS inner
product $\langle f_i, g_j \rangle_{\Nc(\Xc)}$. Furthermore, given a
real coefficient matrix $W \in \real^{p \times r}$, the
right-multiplication $F W$ defines a new collection of $r$ functions
in $\Nc(\Xc)$, where the $j$-th function is given by the linear
combination $\sum_{i=1}^p f_i W_{ij}$.

Utilizing this algebraic notation, given two datasets of points
$X = [x_1 \, x_2 \, \dots \, x_N] \in \real^{n \times N}$ and
$Y = [y_1 \, y_2 \, \dots \, y_M] \in \real^{n \times M}$, we define
the collections of empirical kernel sections
\begin{align}
  \Phi_X \!=\! [\Phi_{x_1} \, \Phi_{x_2} \, \cdots \, \Phi_{x_N}], \label{eqn:kernel_section_defn}
  \,\,
  \Phi_Y\!=\! [\Phi_{y_1} \, \Phi_{y_2} \, \cdots \, \Phi_{y_M}].    
\end{align}
By our established convention, the cross-inner product matrix of these
collections yields the empirical kernel matrix
$K_{X, Y} \triangleq \Phi_X^\top \Phi_Y \in \real^{N \times M}$, whose
elements are
\begin{equation}\label{eqn:kernel_matrix_defn}
  [K_{X, Y}]_{ij} = \langle \Phi_{x_i}, \Phi_{y_j} \rangle_{\Nc(\Xc)} =
  k(x_i, y_j), 
\end{equation}
for $i \in [N], \,\, j \in [M]$.

\subsection{Kernel EDMD as an RKHS Orthogonal Projection}

% \marginJC{But what is the kernel EDMD thing? Have we discussed it?
%   Does it refer to~\cite{MOW-CWR-IGK:14}?}
% \marginDS{Added an introductory paragraph and citing the first paper.}

Kernel Extended Dynamic Mode Decomposition
(kEDMD)~\cite{MOW-CWR-IGK:14} was introduced to alleviate the
computational burden of standard EDMD in high-dimensional systems. By
employing the kernel trick, kEDMD implicitly defines a rich,
high-dimensional dictionary of observables without the need to
explicitly compute them. To derive the kEDMD matrix and analyze this
finite-data approximation rigorously, we adopt the RKHS projection viewpoint~\cite{FK-FMP-MS-AS-KW:25}.

The procedure is the same as described in \eqref{eqn:EDMD_approx}, but
here we work in the RKHS $\mathcal{N}(\mathcal{X})$ induced by the
kernel $k(\cdot, \cdot)$ and use its inner product structure to
compute the orthogonal projection of the Koopman operator.

Given data snapshots $X,T(X) \in \mathbb{R}^{n \times N}$, we define
the data-driven dictionary $\Phi_X$ as the collection of empirical
kernel sections defined in \eqref{eqn:kernel_section_defn}. We
consider the finite-dimensional subspace
$\text{span}(\Phi_X) \subset \mathcal{N}(\mathcal{X})$ spanned by this
dictionary. Any observable function $f \in \text{span}(\Phi_X)$ can be
parameterized by a coefficient vector $\alpha \in \mathbb{R}^N$ such
that $f = \Phi_X \alpha$. The image under the Koopman operator of $f$
is $\mathcal{K}f = \mathcal{K} (\Phi_X \alpha)$. To find its
orthogonal projection
$P_{\text{span}(\Phi_X)} \mathcal{K} f = \Phi_X \beta$, we require the
new coefficient vector $\beta \in \mathbb{R}^N$ to satisfy the
interpolation condition at the nodes $X$ (see \cite[Sec
3.2]{FK-FMP-MS-AS-KW:25}). Evaluating a function at the nodes $X$ in
an RKHS is equivalent to taking the cross-inner product with
$\Phi_X^\top$. Therefore, we equate the evaluations of our projected
function and the true evolved function:
\begin{equation}
  \Phi_X^\top (\Phi_X \beta) = \Phi_X^\top (\mathcal{K}\Phi_X \alpha) .
\end{equation}
On the lefthand side, we have the standard kernel matrix
$K_{X,X} = \Phi_X^\top \Phi_X$. On the righthand side, we have
$\Phi_X^\top \mathcal{K}\Phi_X = K_{T(X),X}$. Substituting these:
\begin{align*}
  K_{X,X} \beta &= K_{T(X),X} \alpha \implies \beta = K_{X,X}^{-1} K_{T(X),X} \alpha.
\end{align*}
Thus, the finite-dimensional kEDMD matrix that advances the coordinate
vector $\alpha$ forward in time, representing the strict orthogonal
projection of the Koopman operator on $\text{span}(\Phi_X)$, is given
by:
\begin{equation}
  K_{\text{kEDMD}} = K_{X,X}^{-1} K_{T(X),X}.
    \label{eqn:kEDMD_matrix}
\end{equation}

% \marginDS{I have gotten into the habit of using vectors and functions
%   interchangeably. I have added it in the notation here. Let me know
%   what you think.}
% %
% \marginJC{Ok for the conf paper, but we'll have to revisit this for
%   the journal version}
%

\section{Problem Formulation}\label{sec:problem_formulation}

Consider the discrete-time dynamical system~\eqref{eqn:sys_dynamics}
and fix a function space $\mathcal{N}(\mathcal{X})$ induced by a
symmetric, positive definite kernel $k(\cdot, \cdot)$. Let
$X, T(X)\in \mathbb{R}^{n \times N}$ denote the trajectory data
snapshots, and let
$\Phi_X = [\Phi_{x_1} \, \Phi_{x_2} \, \cdots \, \Phi_{x_N}]$ denote
the collection of empirical kernel sections.

To systematically analyze the Koopman operator using finite data in a computationally tractable manner, we
restrict our attention to a reduced function space
$\mathcal{S} = \text{span}(\mathcal{V})$, which is spanned by a
collection $\mathcal{V} = [v_1 \, v_2 \, \cdots \, v_s]$ of $s$
observable functions in $\mathcal{N}(\mathcal{X})$. We anchor the small, computationally manageable
$s$-dimensional subspace $\mathcal{S}$ to the rich geometry of the
full dataset by defining $\mathcal{V} = \Phi_X W_{\mathcal{V}}$ for
some coefficient matrix $W_{\mathcal{V}} \in \mathbb{R}^{N \times
  s}$. This ensures that
\begin{align*}
  \mathcal{S} = \text{span}(\mathcal{V}) \subseteq \text{span}(\Phi_X) ,
\end{align*}
where $\dim(\mathcal{S}) \leq s \ll N$. 

Because a
low-dimensional subspace $\mathcal{S}$ is generally not invariant
under the Koopman operator $\mathcal{K}$, we must actively measure and
improve its invariance proximity. This operational challenge motivates
our first problem:

\begin{problem}\longthmtitle{Invariant Subspace Search in
    RKHS}\label{problem:subspace_search}
  Given an initial subspace $\Sc$ and a
  tolerance $\epsilon \in [0, 1)$, find a subspace
  $\mathcal{S}^* \subseteq \Sc$ of the largest
  possible dimension such that:
  \begin{equation}\label{eq:problem-bound}
    \delta(\mathcal{S}^*) \le \epsilon,
  \end{equation}
  where $\delta(\mathcal{S}^*)$ is the invariance proximity of $\mathcal{S}^*$ with respect to the RKHS inner product. \oprocend
\end{problem}

The recently proposed Single Principal Vector (SPV) pruning
procedure~\cite{DS-JC:26-cdc1} provides a systematic solution to this search problem. This iterative method eliminates the principal vector associated with the largest principal angle at each step, followed by a recomputation of the principal arguments for the reduced subspace. This pruning process is repeated until the invariance proximity of the subspace satisfies the prescribed tolerance $\epsilon > 0$.
%
%\marginJC{Briefly explain what SPV consists of: since later we keep
%  referring to the SPV prunning procedure described above, but we
%  never quite described it }
%\marginDS{Resolved}
%

To execute this
procedure, however, one must be able to compute the principal angles
and vectors between the subspaces $\mathcal{S}$ and
$\mathcal{K}\mathcal{S}$. In a standard Euclidean setting, this is
straightforwardly achieved via the singular value decomposition (SVD)
of the basis vectors' inner product matrix. However, computing these
quantities within an RKHS is highly non-trivial due to the underlying
geometry and the implicit nature of its elements. Overcoming this
hurdle is a critical enabler for the SPV procedure and constitutes the
main technical contribution of this paper:

\begin{problem}\longthmtitle{RKHS Principal Angles and
    Vectors}\label{problem:pa_pv}
  Compute the principal angles and principal vectors between the
  subspaces $\mathcal{S}$ and $\mathcal{K}\mathcal{S}$, defined with
  respect to the RKHS inner product
  $\langle \cdot, \cdot \rangle_{\mathcal{N}}$. \oprocend
\end{problem}
\smallskip

The primary challenge in solving Problem \ref{problem:pa_pv} lies in computing the inner products between elements of $\Kc\Sc$. While the kernel trick successfully facilitates these computations within $\Sc$, as well as between $\Sc$ and $\Kc\Sc$, it cannot be applied to inner products exclusively among elements of $\Kc\Sc$. To rigorously formulate the exact computational routine required to overcome this, we must first establish the following formal assumption regarding the system and the entire collected dataset:

% \marginDS{We need this assumption to justify that the non-invariant
% residuals can be made arbitrarily small by more data.}
% \marginDS{Read the remark. I have relaxed the assumptions.}
\begin{assumption}\longthmtitle{Idealized Dataset
    Invariance}\label{assum:ideal_invariance}
  The dataset $X, T(X)$ is sufficiently rich so that the Koopman image
  of the reduced subspace $\mathcal{S}$ is contained within the span
  of the kernel sections defined by the dataset, i.e.,
  $\text{span}(\mathcal{K}\Vc) \subseteq \text{span}(\Phi_X)$ for a
  sufficiently large $N$.
\end{assumption}

Although Assumption~\ref{assum:ideal_invariance} is rarely achieved
with finite data in practice, this idealized condition establishes the
foundation for deriving the principal angles by enabling the
computation of inner products among elements of $\Kc\Sc$ via the
kernel trick.
% We address the practical reality of non-invariant residual errors
% below.

\begin{remark}\longthmtitle{Justification of Finite-Data
    Containment}\label{remark:non_invariant_residuals}
  Although exact subspace containment (cf. Assumption
  \ref{assum:ideal_invariance}) is an idealization in the finite-data
  regime, it is justified by the universality of the
  RKHS. Specifically, the dataset $X$ can be systematically expanded
  to enrich the computational basis $\text{span}(\Phi_X)$ without
  altering the base subspace $\mathcal{S} =
  \text{span}(\mathcal{V})$. Provided the underlying system dynamics
  are sufficiently smooth to preserve the regularity of
  $\text{span}(\Kc \Vc)$, the density of the RKHS ensures that as the
  sample size $N$ grows, the orthogonal residual of the Koopman image
  $\text{span}(\Kc \Vc)$ relative to $\text{span}(\Phi_X)$ can be
  driven arbitrarily close to zero. \oprocend
\end{remark}

\section{Principal Angle and Vector Computation}\label{sec:exact_RKHS_PA}

We derive here the exact computational routine to solve
Problem~\ref{problem:pa_pv}. This procedure involves first
constructing the requisite empirical Gram matrices, utilizing
truncated eigendecompositions for stable
basis orthogonalization, and ultimately computing the principal
arguments via SVD.

Given $\Sc = \text{span}(\Vc) \subset \Nc(\Xc)$, let
$\Kc \Sc = \text{span}(\Kc \Vc)$, where
$\Kc \Vc = [\Kc v_1 \, \Kc v_2 \, \cdots \, \Kc v_s]$.  We can express
$\Kc \Vc$ as a linear combination of the kernel sections in $\Phi_X$.

\begin{lemma}[Computing $W_{\Kc \Vc}$]\label{lemma:W_KV}
  The basis vectors of $\Kc \Sc$ are given by
  $\Kc \Vc = \Phi_X W_{\Kc \Vc}$, where
  $W_{\Kc \Vc} \in \real^{N \times s}$ satisfies
  \begin{equation}\label{eq:W_KV_eqn}
    K_{T(X),X} W_{\Vc} = K_{X,X} W_{\Kc \Vc}.
  \end{equation}
  Here, $K_{T(X),X}, K_{X,X} \in \real^{N \times N}$ are the kernel
  matrices defined according to \eqref{eqn:kernel_matrix_defn}.
\end{lemma}
\begin{proof}
  Based on Assumption \ref{assum:ideal_invariance}, we have
  $\Kc \Sc = \text{span}(\Kc \Vc) \subseteq
  \text{span}(\Phi_X)$. Thus, we can express $\Kc \Vc$ as a linear
  combination of the kernel sections in $\Phi_X$, i.e.,
  $\Kc \Vc = \Phi_X W_{\Kc \Vc}$.  Taking the inner product with
  $\Phi_X$ yields
  \begin{align*}
    \Phi_X^{\top} \Kc \Vc = \Phi_X^{\top} \Kc \Phi_X W_{\Vc} =
    K_{T(X),X} W_{\Vc}
    = \Phi_X^{\top}
    \Phi_X
    W_{\Kc
    \Vc} ,
  \end{align*}
  which implies \eqref{eq:W_KV_eqn}.  The matrix $K_{X,X}$ is positive
  definite because the kernel is positive definite, and this
  guarantees the uniqueness of $W_{\Kc \Vc}$.
\end{proof}

\begin{remark}[Computational Complexity]
  Equation~\eqref{eq:W_KV_eqn} is a linear system of equations that
  can be solved for $W_{\Kc \Vc}$, given $W_{\Vc}$ and the kernel
  matrices $K_{T(X),X}$ and $K_{X,X}$. In practice, we add a small
  regularization term $\lambda I$ to $K_{X,X}$ to ensure numerical
  stability. The computational complexity of solving for $W_{\Kc \Vc}$
  is $\mathcal{O}(N^3)$. This is a significant bottleneck for large
  datasets, motivating the need for scalable approximations tackled in
  the forthcoming Section~\ref{sec:nystrom_approximation}. \oprocend
\end{remark}

Define the Gram matrices associated with the subspaces $\Sc$ and
$\Kc \Sc$:
\begin{equation}
  \begin{aligned}
    M_{\Vc} &= \Vc^{\top} \Vc = W_{\Vc}^{\top} K_{X,X} W_{\Vc} \in
              \real^{s \times s},
    \\ 
    M_{\Kc \Vc} &= (\Kc \Vc)^{\top} (\Kc \Vc) = W_{\Kc \Vc}^{\top}
                  K_{X,X} W_{\Kc \Vc} \in \real^{s \times s},
    \\ 
    M_{\text{cross}} &= \Vc^{\top} (\Kc \Vc) = W_{\Vc}^{\top}
                       K_{T(X),X} W_{\Vc} \in \real^{s \times s}. 
  \end{aligned}
  \label{eqn:Gram_matrices}
\end{equation}
Since we have not made any assumptions on the linear independence of
the vectors in $\Vc$ and $\Kc \Vc$, the Gram matrices may be
rank-deficient. Let $r_{\Vc} = \text{rank}(M_{\Vc}) \le s$ and
$r_{\Kc \Vc} = \text{rank}(M_{\Kc \Vc}) \le r_{\Vc}$. Let their
\texttt{QR} decompositions be given by:
\begin{equation}
  \Vc = Q_{\Vc} R_{\Vc} , \,\, \Kc \Vc = {Q}_{\Kc \Vc} {R}_{\Kc \Vc}
\label{eqn:QR_decomp}
\end{equation}
where $Q_{\Vc}$ and $Q_{\Kc \Vc}$ are a collection of orthonormal
vectors with $r_{\Vc}$ and $r_{\Kc \Vc}$ vectors respectively, and
$R_{\Vc} \in \real^{r_{\Vc} \times s}$ and
$R_{\Kc \Vc} \in \real^{r_{\Kc \Vc} \times s}$ are upper triangular
matrices.

Note that we cannot explicitly compute the \texttt{QR} decompositions
of $\Vc$ and $\Kc \Vc$ since they are represented implicitly as linear
combinations of kernel sections. However, we can obtain the
\texttt{QR} decompositions from the eigen decompositions of
the Gram matrices $M_{\Vc}$ and $M_{\Kc \Vc}$, as described next.

\begin{lemma}\longthmtitle{Obtaining QR Decompositions via Spectral Truncation}\label{lemma:qr_decompositions}
  Let $M_{\mathcal{V}}$ and $M_{\mathcal{K}\mathcal{V}}$ be the symmetric, positive semi-definite Gram matrices, and consider their truncated eigendecompositions:
  \begin{equation}\label{eqn:eigen_decomp}
    \begin{aligned}
      M_{\mathcal{V}} = \tilde{V}_{\mathcal{V}} \tilde{\Lambda}_{\mathcal{V}} \tilde{V}_{\mathcal{V}}^{\top} , \,\,
      M_{\mathcal{K}\mathcal{V}} = \tilde{V}_{\mathcal{K}\mathcal{V}} \tilde{\Lambda}_{\mathcal{K}\mathcal{V}} \tilde{V}_{\mathcal{K}\mathcal{V}}^{\top} ,
    \end{aligned}
  \end{equation}
  where $\tilde{\Lambda}_{\mathcal{V}} \in \mathbb{R}^{r_{\mathcal{V}} \times r_{\mathcal{V}}}$ and $\tilde{\Lambda}_{\mathcal{K}\mathcal{V}} \in \mathbb{R}^{r_{\mathcal{K}\mathcal{V}} \times r_{\mathcal{K}\mathcal{V}}}$ are diagonal matrices containing the strictly positive eigenvalues. Consequently, $r_{\mathcal{V}}$ and $r_{\mathcal{K}\mathcal{V}}$ denote the exact ranks of $M_{\mathcal{V}}$ and $M_{\mathcal{K}\mathcal{V}}$, respectively. Let $\tilde{V}_{\mathcal{V}}, \tilde{V}_{\mathcal{K}\mathcal{V}}$ contain the corresponding eigenvectors. Define the intermediate matrices $B_{\mathcal{V}} = \tilde{\Lambda}_{\mathcal{V}}^{1/2} \tilde{V}_{\mathcal{V}}^{\top}$ and $B_{\mathcal{K}\mathcal{V}} = \tilde{\Lambda}_{\mathcal{K}\mathcal{V}}^{1/2} \tilde{V}_{\mathcal{K}\mathcal{V}}^{\top}$. Let $B_{\mathcal{V}} = Q_{B_{\mathcal{V}}} R_{\mathcal{V}}$ and $B_{\mathcal{K}\mathcal{V}} = Q_{B_{\mathcal{K}\mathcal{V}}} R_{\mathcal{K}\mathcal{V}}$ be their respective economic \texttt{QR} decompositions. Then, the upper-triangular factors $R_{\mathcal{V}}$ and $R_{\mathcal{K}\mathcal{V}}$ satisfy $M_{\mathcal{V}} = R_{\mathcal{V}}^\top R_{\mathcal{V}}$ and $M_{\mathcal{K}\mathcal{V}} = R_{\mathcal{K}\mathcal{V}}^\top R_{\mathcal{K}\mathcal{V}}$, and 
  \begin{equation}
    \begin{aligned}
      Q_{\mathcal{V}} &= \mathcal{V} R_{\mathcal{V}}^{\dagger},  \quad R_{\mathcal{V}}^{\dagger} = \tilde{V}_{\mathcal{V}} \tilde{\Lambda}_{\mathcal{V}}^{-1/2} Q_{B_{\mathcal{V}}}, \\
      Q_{\mathcal{K} \mathcal{V}} &= \mathcal{K} \mathcal{V} R_{\mathcal{K} \mathcal{V}}^{\dagger},  \quad R_{\mathcal{K}\mathcal{V}}^{\dagger} = \tilde{V}_{\mathcal{K}\mathcal{V}} \tilde{\Lambda}_{\mathcal{K}\mathcal{V}}^{-1/2} Q_{B_{\mathcal{K}\mathcal{V}}}, 
    \end{aligned}
    \label{eqn:R_terms_eigen}
  \end{equation}
  describe the orthogonal \texttt{QR} factors of the implicit subspaces.
\end{lemma}
\smallskip

\begin{proof}
  We verify the procedure for $\mathcal{V}$; the proof for $\mathcal{K}\mathcal{V}$ is identical. Because the Gram matrix $M_{\mathcal{V}}$ is symmetric and positive semi-definite, its rank is exactly equal to the number of its strictly positive eigenvalues, confirming $\text{rank}(M_{\mathcal{V}}) = r_{\mathcal{V}}$. 
  Next, note that $R_{\mathcal{V}}^\top R_{\mathcal{V}} = (Q_{B_{\mathcal{V}}}^\top B_{\mathcal{V}})^\top (Q_{B_{\mathcal{V}}}^\top B_{\mathcal{V}}) = B_{\mathcal{V}}^\top B_{\mathcal{V}} = \tilde{V}_{\mathcal{V}} \tilde{\Lambda}_{\mathcal{V}} \tilde{V}_{\mathcal{V}}^\top$, which recovers the truncated Gram matrix. Since $Q_{B_{\mathcal{V}}}$ is a square orthogonal matrix, the Moore-Penrose pseudo-inverse of $R_{\mathcal{V}} = Q_{B_{\mathcal{V}}}^\top B_{\mathcal{V}}$ is explicitly given by $R_{\mathcal{V}}^\dagger = B_{\mathcal{V}}^\dagger Q_{B_{\mathcal{V}}} = \tilde{V}_{\mathcal{V}} \tilde{\Lambda}_{\mathcal{V}}^{-1/2} Q_{B_{\mathcal{V}}}$. Finally,
  \begin{equation*}
    Q_{\mathcal{V}}^\top Q_{\mathcal{V}} = (R_{\mathcal{V}}^{\dagger})^{\top} M_{\mathcal{V}} R_{\mathcal{V}}^{\dagger} = (R_{\mathcal{V}}^{\dagger})^{\top} (R_{\mathcal{V}}^{\top} R_{\mathcal{V}}) R_{\mathcal{V}}^{\dagger} = I_{r_{\mathcal{V}}},
  \end{equation*}
  verifies that $Q_{\mathcal{V}}$ is orthonormal.
\end{proof}

\smallskip
\begin{remark}\longthmtitle{Numerical
    Implementation}\label{remark:numerical_tolerance} 
  Because Gram matrices often exhibit exponentially decaying spectra,
  exact rank determination is computationally fragile. In practice, we
  implement the eigendecompositions in~\eqref{eqn:eigen_decomp} using
  a hard numerical tolerance (e.g., $\tau = 10^{-8}$), discarding
  eigenvalues below this threshold to safely compute the inverse
  square roots without inflating numerical noise. \oprocend
\end{remark}

We utilize the \texttt{QR} decompositions to compute the principal
angles and vectors via \cite[Theorem 1]{AB-GHG:73}. Let
$\{u_i^\Sc\}_{i=1}^{k_r} \subset \Sc$ and
$\{\Kc v_i^{\Kc \Sc}\}_{i=1}^{k_r} \subset \Kc \Sc$ be the principal
vectors corresponding to the principal angles $\{\theta_i\}_{i=1}^{k_r}$
between $\Sc$ and $\Kc \Sc$. Here, $k_r = \min(r_{\Vc}, r_{\Kc \Vc})$ is
the number of principal angles and vectors. The next result gives an
explicit formula for computing these principal arguments wrt the RKHS.

\smallskip
\begin{theorem}\longthmtitle{Principal Angles and Vectors in
    RKHS}\label{thm:pa_pv_rkhs}
  Let the cosine matrix be defined as
  $\real^{r_{\Vc} \times r_{\Kc \Vc}} \ni M = Q_{\Vc}^{\top} {Q}_{\Kc
    \Vc} = (R_{\Vc}^{\dagger})^{\top} M_{\text{cross}} {R}_{\Kc
    \Vc}^\dagger$. Let the SVD of $M$ be given by
  ${U} \Sigma {V}^{\top} = M$, where $U \in \real^{r_{\Vc} \times k_r}$
  and $V \in \real^{r_{\Kc \Vc} \times k_r}$ are orthonormal matrices,
  and $\Sigma = \text{diag}(\sigma_1, \dots, \sigma_{k_r})$ with
  $\sigma_1 \geq \cdots \geq \sigma_{k_r} \geq 0$. Then, the principal
  angles and vectors between $\Sc$ and $\Kc \Sc$ for $i = 1, \dots, k_r$ are given by
  \begin{equation}
    \begin{aligned}
      \cos \theta_i = \sigma_i, \,\,
      \Uc^{\Sc} = \Vc \Ac_{\Vc}, \,\,
      \Kc \Vc^{\Kc \Sc} = \Kc \Vc \Ac_{\Kc \Vc},
    \end{aligned}
  \end{equation}
  %
  % \marginJC{What are $\Ac_{\Vc}$ and $\Ac_{\Kc \Vc}$? Or are you
  % introducing them here?}  \marginDS{They are the coefficient
  % vectors. I had an error %before, it is corrected now.}
  %
  where $\Ac_{\Vc} = R_{\Vc}^{\dagger} U$ and $\Ac_{\Kc \Vc} =
  R_{\Kc \Vc}^{\dagger} V$ are the coordinate matrices of the principal vectors
  $\Uc^{\Sc} \!=\! [u_1^\Sc \, \cdots \, u_{k_r}^\Sc]$ and
  $\Kc \Vc^{\Kc \Sc} \!=\! [\Kc v_1^{\Kc \Sc} \, \cdots \, \Kc
  v_{k_r}^{\Kc \Sc}]$.
\end{theorem}
\smallskip
\begin{proof}
  Because $Q_{\Vc} = \Vc R_{\Vc}^{\dagger}$ and
  $Q_{\Kc\Vc} = \Kc \Vc R_{\Kc \Vc}^{\dagger}$ are orthonormal bases
  for $\Sc$ and $\Kc\Sc$, any functions $u \in \Sc$ and $v \in \Kc\Sc$
  can be uniquely parameterized as $u = Q_{\Vc} \alpha$ and
  $v = Q_{\Kc\Vc} \beta$, for some coordinate vectors
  $\alpha \in \real^{r_{\Vc}}$ and $\beta \in
  \real^{r_{\Kc\Vc}}$.
  Crucially, the orthonormality of these bases establishes an isomorphism between the function subspaces and their respective Euclidean coordinate spaces, ensuring that $\|u\|_{\mathcal{N}} = \|\alpha\|_2$ and $\|v\|_{\mathcal{N}} = \|\beta\|_2$.
  %
  %\marginJC{Here you've fully changed from the functions to %vectors
  %  terminology, but I don't know that we've been explicit %about the
  %  isomorphism with the reader earlier.}
  %\marginDS{The function vs vector dilemma has been addressed %in the notation section. I have explicitly added how the %isomorphisms are defined now.}
  %
  Computing the principal angles requires sequentially maximizing the
  inner product $\langle u, v \rangle_{\Nc}$ subject to unit norm and
  orthogonality constraints,
  cf. Section~\ref{sec:principal-angles}. Substituting the coordinate
  representations, 
  $$
  \langle u, v \rangle_{\Nc} = \langle Q_{\Vc} \alpha, Q_{\Kc\Vc}
  \beta \rangle_{\Nc} = \alpha^{\top} M \beta.
  $$ 
  Therefore, the infinite-dimensional optimization over $\Sc$ and
  $\Kc\Sc$ is mathematically equivalent to:
  $$
  \max_{\alpha, \beta} \alpha^{\top} M \beta \quad \text{subject to}
  \quad \|\alpha\|_2 = 1, \|\beta\|_2 = 1,
  $$
  with subsequent orthogonality constraints
  $\alpha_i^\top \alpha_j = \delta_{ij}$ and
  $\beta_i^\top \beta_j = \delta_{ij}$. By
  \cite[Theorem~1]{AB-GHG:73}, the optimal coordinate vectors solving
  this sequence of Euclidean problems are exactly the columns of $U$
  and $V$ from the SVD $M = U \Sigma V^\top$, and the maximized inner
  product values are the singular values $\cos \theta_i =
  \sigma_i$. Mapping these optimal coordinates back into the RKHS
  yields the functional principal vectors $\Uc^{\Sc} = Q_{\Vc} U$ and
  $\Kc \Vc^{\Kc \Sc} = Q_{\Kc \Vc} V$.
\end{proof}
\smallskip

\begin{algorithm}[htb]
  \caption{Exact Principal Angles and Vectors in RKHS}
  \label{alg:exact_kernel_pa_pv}
  \begin{algorithmic}[1]
    \REQUIRE Data matrices $X, T(X) \in \mathbb{R}^{n \times N}$, Dictionary coefficient matrix $W_{\mathcal{V}} \in \mathbb{R}^{N \times s}$, Kernel function $k(\cdot, \cdot)$
    \ENSURE Exact principal angles $\{\theta_i\}_{i=1}^{k_r}$, Principal vector
    coefficient matrices $A_{\mathcal{V}}, A_{\mathcal{K}\mathcal{V}}$ 
    
    \vspace{0.1cm}
    
    \STATE Compute $K_{X,X}, K_{T(X),X} \in \mathbb{R}^{N \times N}$
    using~\eqref{eqn:kernel_matrix_defn}
    
    \vspace{0.1cm}
    
    \STATE Compute
    $W_{\mathcal{K}\mathcal{V}} \in \mathbb{R}^{N \times s}$ by
    solving \eqref{eq:W_KV_eqn}
    
    \vspace{0.1cm}
    
    \STATE Construct the Gram matrices
    $M_{\mathcal{V}}, M_{\mathcal{K}\mathcal{V}}, M_{\text{cross}}$ in
    \eqref{eqn:Gram_matrices}
    
    \vspace{0.1cm}
    
    \STATE Compute the rank-revealing eigendecompositions
    \eqref{eqn:eigen_decomp} 
    
    \vspace{0.1cm}
    
    \STATE Form the \texttt{QR} factors $R_{\mathcal{V}}$ and
    $R_{\mathcal{K}\mathcal{V}}$ using \eqref{eqn:R_terms_eigen} 
    \vspace{0.1cm}
    
    \STATE Compute the cosine matrix $M =
    (R_{\mathcal{V}}^{\dagger})^\top M_{\text{cross}}
    R_{\mathcal{K}\mathcal{V}}^\dagger$ 
    
    \vspace{0.1cm}
    
    \STATE Compute the \texttt{SVD} of the cosine matrix: $M = U \Sigma V^\top$
    
    \vspace{0.1cm}
    
    \STATE Extract the principal arguments using Theorem
    \ref{thm:pa_pv_rkhs}:
    \\
    \hspace{0.2cm} $\cos \theta_i = \Sigma_{i,i}$ 
    \hspace{0.2cm} $A_{\mathcal{V}} = R_{\mathcal{V}}^\dagger U$ 
    \hspace{0.2cm} $A_{\mathcal{K}\mathcal{V}} = R_{\mathcal{K}\mathcal{V}}^\dagger V$
    
    \vspace{0.1cm}       
    
    \RETURN $\{\theta_i\}_{i=1}^{k_r}$, $A_{\mathcal{V}}$, $A_{\mathcal{K}\mathcal{V}}$
  \end{algorithmic}
\end{algorithm}
%
%\marginJC{$i$ in step 8 runs from 1 to $k$? The current th %statement
%  does not quite express $A_{\mathcal{V}}$,
%  $A_{\mathcal{K}\mathcal{V}}$ this explicitly.}
%\marginDS{Resolved}
%

%%
%\marginJC{I'm confused, are the principal vectors $A_{\mathcal{V}}$,
%  $A_{\mathcal{K}\mathcal{V}}$, which is what the algorithm seems to
%  say, or $\Uc^{\Sc}$ and $\Kc \Vc^{\Kc \Sc}$, which is what the
%  Theorem says??}
%\marginDS{Resolved. Had a problem in the Theorem statement. %Fixed now.}  
%
Algorithm~\ref{alg:exact_kernel_pa_pv} summarizes the overall
procedure for computing the exact principal angles and vectors in the
RKHS.  The computational complexity of
Algorithm~\ref{alg:exact_kernel_pa_pv} is dominated by solving for
$W_{\Kc \Vc}$ in Step 2, which is $\mathcal{O}(N^3)$, and this
motivates the need for scalable approximations that we discuss in the
next section.

\section{Scalable Computation via the Nystr\"om Approximation}\label{sec:nystrom_approximation}

To overcome the $\mathcal{O}(N^3)$ computational bottleneck associated with the dense kernel matrix $K_{X,X}$, we employ the Nystr\"om approximation \cite{CW-MS:00}. The Nystr\"om method provides a data-dependent, low-rank approximation of the Gram matrix that inherently adapts to the geometry of the underlying state manifold, thereby improving the sample efficiency of the approximation. 

The procedure begins by sampling a small subset of $D \ll N$ ``landmark'' data points from the full trajectory dataset. These can be uniformly sampled to capture the diversity of the state space, or selected via clustering methods to identify representative states. Let $L = [l_1 \, l_2 \, \cdots \, l_D] \in \mathbb{R}^{n \times D}$ denote this collection of landmarks. We then partition the kernel evaluations to define the landmark kernel matrix $K_{L,L} \in \mathbb{R}^{D \times D}$ and the cross-kernel matrix $K_{L,X} \in \mathbb{R}^{D \times N}$, which relates the landmarks to the entire dataset. The Nystr\"om method approximates the exact kernel matrix as:
\begin{equation}
  K_{X,X} \approx \tilde{K}_{X,X} = K_{X,L} K_{L,L}^{\dagger} K_{L,X}.
\end{equation}

To seamlessly integrate this low-rank approximation into our proposed subspace analysis, we construct an explicit, finite-dimensional feature map $\psi: \mathbb{R}^n \to \mathbb{R}^D$. Consider the symmetric eigendecomposition of the strictly positive definite landmark matrix:
\begin{equation}
    K_{L,L} = U_L \Lambda_L U_L^\top ,
\end{equation}
where $\Lambda_L \in \mathbb{R}^{D \times D}$ is the diagonal matrix of positive eigenvalues, and $U_L \in \mathbb{R}^{D \times D}$ contains the corresponding orthonormal eigenvectors. Using these spectral factors, we define the explicit feature map $\psi(x)$ for any state $x \in \mathcal{X}$ as:
\begin{equation}\label{eq:nystrom_map}
  \psi(x) = \Lambda_L^{-1/2} U_L^\top k_L(x) \in \mathbb{R}^{D} ,
\end{equation}
where $k_L(x) = [k(l_1, x), \, \dots, \, k(l_D, x)]^\top \in \mathbb{R}^D$ is the column vector of kernel evaluations between the state $x$ and the landmark points. 

By applying this data-driven mapping to the entire trajectory dataset $X$, we form the explicit data matrix $\Psi(X) = [\psi(x_1) \, \psi(x_2) \, \cdots \, \psi(x_N)]$. This matrix can be computed efficiently via a direct matrix multiplication:
\begin{equation}
    \Psi(X) = \Lambda_L^{-1/2} U_L^\top K_{L,X} \in \mathbb{R}^{D \times N} .
\end{equation}
Consequently, the exact $N \times N$ kernel matrix is approximated by the inner product of our new feature matrices, i.e. $K_{X,X} \approx \Psi(X)^{\top} \Psi(X)$. As $D$ increases, this approximation converges to the true kernel matrix. The error bounds for the Nystr\"om approximation are well-established (cf. \cite{AG-MM:13}), and the method performs particularly well when the exact kernel matrix exhibits rapidly decaying eigenvalues---a condition frequently met in practice due to the smoothness of the underlying dynamics and the selected kernel.

Although this approach introduces an approximation error relative to the true RKHS inner product $\langle \cdot, \cdot \rangle_{\mathcal{N}}$, it elegantly resolves the scalability challenge associated with constructing and inverting the full kernel matrix $K_{X,X}$. Recall that in Section \ref{sec:exact_RKHS_PA}, computing $W_{\Kc \Vc}$ created a significant computational bottleneck. By employing the Nystr\"om feature map, we can instead operate directly on the explicit feature matrices $\Psi(X), \Psi(T(X)) \in \mathbb{R}^{D \times N}$, entirely bypassing the computation of $W_{\Kc \Vc}$.

We adapt our principal angle and vector computations to this new feature space through the following procedure. First, we construct finite-dimensional target matrices $Z_{\Vc}, Z_{\Kc \Vc} \in \mathbb{R}^{D \times s}$ designed to approximate the true Gram matrices such that $M_{\Vc} \approx Z_{\Vc}^\top Z_{\Vc}$ and $M_{\Kc \Vc} \approx Z_{\Kc \Vc}^\top Z_{\Kc \Vc}$. Performing rank-revealing \texttt{QR} decompositions on these targets, $Z_{\Vc} = Q_{Z_{\Vc}} \tilde{R}_{\Vc}$ and $Z_{\Kc \Vc} = Q_{Z_{\Kc \Vc}} \tilde{R}_{\Kc \Vc}$, yields the approximate factorizations $M_{\Vc} \approx \tilde{R}_{{\Vc}}^\top \tilde{R}_{\Vc}$ and $M_{\Kc \Vc} \approx \tilde{R}_{\Kc \Vc}^\top \tilde{R}_{\Kc \Vc}$. Applying Lemma \ref{lemma:qr_decompositions}, we then establish the approximate orthonormal bases $\tilde{Q}_{\Vc} = \Vc \tilde{R}_{\Vc}^{\dagger}$ and $\tilde{Q}_{\Kc \Vc} = \Kc \Vc \tilde{R}_{\Kc \Vc}^{\dagger}$. Finally, mirroring Theorem \ref{thm:pa_pv_rkhs}, the principal angles and vectors are extracted via the SVD of the approximate cosine matrix $\tilde{M} = \tilde{Q}_{\Vc}^\top \tilde{Q}_{\Kc \Vc}$.

The success of this procedure hinges on two critical steps: accurately constructing the target matrices $Z_{\Vc}$ and $Z_{\Kc \Vc}$, and correctly matching the rank of the truncated \texttt{QR} decompositions to the true ranks of the underlying Gram matrices. We address the former through a specialized target matrix design, and the latter via data-dependent singular value thresholding, both detailed in the subsequent subsections.

\subsection{Designing Target Matrices for Nystr\"om Approximation}

For the base dictionary $\mathcal{V}$, the target matrix is straightforwardly defined as $Z_{\mathcal{V}} = \Psi(X) W_{\mathcal{V}}$. By construction, this matrix satisfies $Z_{\mathcal{V}}^\top Z_{\mathcal{V}} = W_{\mathcal{V}}^\top \Psi(X)^\top \Psi(X) W_{\mathcal{V}} \approx W_{\Vc}^{\top} K_{X,X} W_{\Vc} = M_{\Vc}$.

Designing the target matrix $Z_{\Kc \Vc}$ for the Koopman image subspace $\mathcal{K}\mathcal{V}$ is more subtle. Our objective is twofold: to approximate the Gram matrix $M_{\Kc \Vc}$ without explicitly computing $W_{\Kc \Vc}$ via the full kernel matrix, and to ensure robustness against the potential rank deficiency of $M_{\Kc \Vc}$. Recall that the exact coefficient matrix $W_{\Kc \Vc}$ satisfies:
\begin{align*}
K_{X,X} W_{\Kc \Vc} &= K_{T(X),X} W_{\Vc}\\
\Rightarrow \Psi(X)^{\top} \Psi(X) W_{\Kc \Vc} &\approx \Psi(T(X)) ^{\top} \Psi(X) W_{\Vc}.
\end{align*}
Instead of solving for $W_{\Kc \Vc}$ and subsequently forming the product $\Psi(X) W_{\Kc \Vc}$, we define $Z_{\Kc \Vc}$ directly as the Tikhonov-regularized solution \cite{ANT-AVY:77} to the approximate least-squares problem $\Psi(X)^{\top} Z_{\Kc \Vc} \approx \Psi(T(X)) ^{\top} \Psi(X) W_{\Vc}$. Substituting $Z_{\Vc} = \Psi(X) W_{\Vc}$, this yields:
\begin{align}\label{eqn:Z_KV_tikhonov}
Z_{\Kc \Vc} \!=\! (\Psi(X) \Psi(X)^{\top} \!+\! \lambda I)^{-1} \Psi(X)\Psi(T(X)) ^{\top} Z_{\Vc},
\end{align}
where $\lambda > 0$ is a regularization parameter that we will discuss
shortly.  The regularized formulation \eqref{eqn:Z_KV_tikhonov} is a
critical design choice that addresses both computational efficiency
and numerical stability. First, it only requires inverting
$\Psi(X) \Psi(X)^{\top} \in \real^{D \times D}$, elegantly bypassing
the massive computational bottleneck of inverting the full kernel
matrix $K_{X,X}$. Second, the regularization parameter $\lambda > 0$
stabilizes the inversion against the inherent rank deficiency of the
subspace. By damping near-zero singular values, Tikhonov
regularization prevents the artificial amplification of numerical
noise, ensuring that the inner product $Z_{\Kc \Vc}^\top Z_{\Kc \Vc}$
provides a robust and reliable geometric approximation of the true
Gram matrix $M_{\Kc \Vc}$.
%
% \marginJC{I'm not familiar with this ``stabilized'' terminology. Is
% it coming from numerical analysis?}  \marginDS{Yes. Basically, we do
% not amplify numerical noise. If we did not
% add %the Tikhonov regularizer, we would be amplifying noise because $M_KV$ is not %usually full rank.}
%

With the target matrices $Z_{\Vc}$ and $Z_{\Kc \Vc}$ defined, we must
address the rank mismatch inherent to the Nystr\"om
approximation. While the true Gram matrices $M_{\Vc}$ and
$M_{\Kc \Vc}$ possess exact mathematical ranks $r_{\Vc}$ and
$r_{\Kc \Vc}$, their empirical counterparts $Z_{\Vc}$ and
$Z_{\Kc \Vc}$ are generically full rank. For $Z_{\Vc}$ and
$Z_{\Kc \Vc}$, this rank inflation occurs because the Nystr\"om method
approximates the geometry of the full dataset using only $D$
landmarks, leaving an approximation residual when extrapolating the
feature map to the remaining data via
\eqref{eq:nystrom_map}. Furthermore for $Z_{\Kc \Vc}$, the Tikhonov
regularization ($\lambda I$) explicitly perturbs the spectrum to
ensure inversion stability.

To prevent the subsequent \texttt{QR} decompositions from spanning these artifactual dimensions---which would severely corrupt the principal angle computations---we apply rank-separation thresholds via truncated SVD \cite{GHG-CFVL:13}. This involves performing SVDs of the target matrices $Z_{\Vc}$ and $Z_{\Kc \Vc}$, and retaining only the singular values that exceed data-dependent thresholds $\tau_{\Vc}(D)$ and $\tau_{\Kc \Vc}(D)$, respectively. This yields
\begin{subequations}\label{eqn:trunc_SVD_Z_V_Z_KV}
  \begin{align}
    Z_{\mathcal{V}}^{\mathrm{trunc}}
    &=
      U_{Z_{\mathcal{V}}}^{\mathrm{trunc}}\,
      \Sigma_{Z_{\mathcal{V}}}^{\mathrm{trunc}}
      \,
      (V_{Z_{\mathcal{V}}}^{\mathrm{trunc}})^\top ,
    \\
    Z_{\mathcal{K}\mathcal{V}}^{\mathrm{trunc}}
    &=
      U_{Z_{\mathcal{K}\mathcal{V}}}^{\mathrm{trunc}}
      \,
      \Sigma_{Z_{\mathcal{K}\mathcal{V}}}^{\mathrm{trunc}}
      \,
      (V_{Z_{\mathcal{K}\mathcal{V}}}^{\mathrm{trunc}})^\top  .
  \end{align}
\end{subequations}
In practice, these thresholds act as tuning parameters set to constant
multiples of $D^{-1/2}$, depending on the decay rate of the kernel
spectrum. For the remainder of this paper, we assume these thresholds
are appropriately tuned such that the ranks of the truncated target
matrices perfectly match the true ranks of their respective Gram
matrices.%
%\marginJC{Some refs for terminology (``stabilized'', ``elbow'') would
%  be helpful here.}
%\marginDS{This is basically plotting the singular values and looking at a clear separation. However, in my simulations, that is insufficent. It's more of a trial and error game. For this work, I am leaving it at "these need to be tuned".}
%anchor

\subsection{Approximate Principal Angles and Vectors}

Next, we extract the upper-triangular matrices $\tilde{R}_{\Vc} = \Sigma_{Z_{\mathcal{V}}}^{\mathrm{trunc}} (V_{Z_{\mathcal{V}}}^{\mathrm{trunc}})^\top$ and $\tilde{R}_{\Kc \Vc} = \Sigma_{Z_{\mathcal{K}\mathcal{V}}}^{\mathrm{trunc}} (V_{Z_{\mathcal{K}\mathcal{V}}}^{\mathrm{trunc}})^\top$ from the truncated SVDs of the target matrices in \eqref{eqn:trunc_SVD_Z_V_Z_KV}. These stable factors act as regularized \texttt{QR} components. They allow us to define approximately orthonormal bases for our subspaces:
\[
  \tilde{Q}_{\mathcal{V}} = \mathcal{V} \tilde{R}_{\Vc}^\dagger, \quad 
  \tilde{Q}_{\mathcal{K}\mathcal{V}} = \mathcal{K}\mathcal{V} \tilde{R}_{\Kc \Vc}^\dagger.
\]

Crucially, we do not need to explicitly compute $\tilde{Q}_{\mathcal{K}\mathcal{V}}$. This would be impossible since the true Koopman image $\mathcal{K}\mathcal{V}$ is unknown. Instead, this implicit representation enables us to compute all necessary inner products algebraically.

Finally, we compute the principal angles and vectors that characterize the invariance proximity of the Koopman subspace. We first form the approximate cross-Gram matrix $\tilde{M}_{\mathrm{cross}} = W_{\mathcal{V}}^\top \Psi(T(X))^\top \Psi(X) W_{\mathcal{V}}$. Following Theorem \ref{thm:pa_pv_rkhs}, this allows us to construct the computable cosine matrix entirely from known quantities:
\begin{equation}
  \tilde{M} = \left( \tilde{R}_{\Vc}^\dagger \right)^\top \tilde{M}_{\mathrm{cross}} \tilde{R}_{\Kc \Vc}^\dagger.
  \label{eqn:M_tilde_defn}  
\end{equation}

By computing the SVD of $\tilde{M} = \tilde{U} \tilde{\Sigma} \tilde{V}^\top$, we extract the approximate principal angles $\cos \tilde{\theta}_i = \tilde{\Sigma}_{i,i}$ for $i = 1, \ldots, \tilde{k}_r$. Furthermore, the approximate principal vectors $\tilde{\Uc}^{\Sc} = \Vc \tilde{\Ac}_{\Vc}$ and $\Kc \tilde{\Vc}^{\Kc \Sc} = \Kc \Vc \tilde{\Ac}_{\Kc \Vc}$ are parameterized by the explicit coefficient matrices:
\begin{equation}
  \tilde{A}_{\mathcal{V}} = (\tilde{R}_{\Vc})^{\dagger} \tilde{U}, \quad
  \tilde{A}_{\mathcal{K}\mathcal{V}} = (\tilde{R}_{\Kc \Vc})^{\dagger} \tilde{V},
  \label{eqn:approx_principal_vecs}
\end{equation}
which efficiently reconstruct the principal vectors in the coordinates of the original dictionaries.

%
%\marginJC{Same question I had earlier for Algo~2. Which ones are the
%  principal vectors, $ \tilde{A}_{\mathcal{V}}$ or
%  $\tilde{\Uc}^{\Sc}$?}
%\marginDS{They are $\tilde{\Uc}^{\Sc}$ and $\Kc \tilde{\Vc}^{\Kc \Sc}%$. The $\tilde{A}$ matrices are just the coefficients for the principal %vectors in the original dictionary coordinates. I have updated the statement %to make this more clear.}
%

Algorithm~\ref{alg:approx_kernel_pa_pv} summarizes the complete Nystr\"om-based computational routine. By operating entirely within the explicit feature space rather than evaluating exact RKHS inner products, this approach provides a highly scalable and numerically robust mechanism for evaluating Koopman subspace pruning geometries on massive trajectory datasets.

%anchor
\begin{algorithm}[htb]
  \caption{Approximate Principal Angles and Vectors}
  \label{alg:approx_kernel_pa_pv}
  \begin{algorithmic}[1]
    \REQUIRE Data matrices $X, T(X) \in \mathbb{R}^{n \times N}$,
    Subspace dictionary coefficient matrix $W_{\mathcal{V}}$, Number
    of random features $D$, Thresholds
    $\tau_{\mathcal{V}}(D),
    \tau_{\mathcal{K}\mathcal{V}}(D)$
    \ENSURE Approximate principal angles $\{\tilde{\theta}_i\}_{i=1}^{\tilde{k}_r}$,
    Approximate principal vector coefficient matrices
    $\tilde{A}_{\mathcal{V}},
    \tilde{A}_{\mathcal{K}\mathcal{V}}$

    \vspace{0.1cm}
    
    \STATE Compute feature matrices $\Psi(X), \Psi(T(X)) \in \mathbb{R}^{D \times N}$
    
    \vspace{0.1cm}
    
    \STATE Compute the target matrices $Z_{\mathcal{V}}, Z_{\mathcal{K}\mathcal{V}}$
            
    \vspace{0.1cm}
    
    \STATE $Z^{\mathrm{trunc}}_{\mathcal{V}}$ $\gets$ T-SVD of 
    $Z_{\mathcal{V}}$ with $\tau_{\mathcal{V}}(D)$
    
    \vspace{0.1cm}
    
    \STATE $Z^{\mathrm{trunc}}_{\Kc \mathcal{V}}$ $\gets$ T-SVD of 
    $Z_{\mathcal{K}\mathcal{V}}$ with $\tau_{\mathcal{K}\mathcal{V}}(D)$
    
    \vspace{0.1cm}
    
    \STATE Compute $\tilde{R}_{\Vc}$ and $\tilde{R}_{\Kc \Vc}$ according to \eqref{eqn:trunc_SVD_Z_V_Z_KV}
    
    \vspace{0.1cm}
    
    \STATE Compute $\tilde{M}_{\mathrm{cross}} = W_{\mathcal{V}}^\top \Psi(T(X))^\top \Psi(X) W_{\mathcal{V}}$.
    
    \vspace{0.1cm}
    
    \STATE Compute SVD of \eqref{eqn:M_tilde_defn}: $\tilde{M} = \tilde{U} \tilde{\Sigma} \tilde{V}^\top$
    
    \vspace{0.1cm}
    
    \STATE $\cos \tilde{\theta}_i = \tilde{\Sigma}_{i,i}$, 
    $\tilde{A}_{\mathcal{V}} = (\tilde{R}_{\Vc})^{\dagger} \tilde{U}$, $\tilde{A}_{\mathcal{K}\mathcal{V}} = (\tilde{R}_{\Kc \Vc})^{\dagger} \tilde{V}$
    
    \vspace{0.1cm}       
    
    \RETURN $\{\tilde{\theta}_i\}_{i=1}^{\tilde{k}_r}$, $\tilde{A}_{\mathcal{V}}$, $\tilde{A}_{\mathcal{K}\mathcal{V}}$
  \end{algorithmic}
\end{algorithm}

\begin{remark}\longthmtitle{Computational Complexity of Algorithm~\ref{alg:approx_kernel_pa_pv}}
  The computational cost of Algorithm~\ref{alg:approx_kernel_pa_pv} is
  dominated by the $\mathcal{O}(D^3)$ eigendecompositions and
  inversions of the landmark matrices, alongside the
  $\mathcal{O}(NDs)$ dense matrix multiplications required to
  construct the feature targets. Assuming $s \le D \ll N$, the overall
  asymptotic complexity is strictly bounded by
  $\mathcal{O}(NDs + D^3)$. This Nystr\"om-based dimensionality
  reduction effectively replaces the prohibitive $\mathcal{O}(N^3)$
  scaling of the exact formulation with a highly tractable linear
  dependence on the dataset size $N$. \oprocend
\end{remark}

\section{Simulation Results}

Having established the procedures for computing principal arguments
within the RKHS, we are now positioned to implement the SPV pruning
strategy. Hereafter, we refer to the exact computational routine as
Kernel-SPV and its Nystr\"om-based counterpart as Approximate
Kernel-SPV.  In this section, we validate the computation of
approximate principal arguments and demonstrate the efficacy of our
pruning strategies using the damped Duffing oscillator. Simulations
were performed in Python 3.11.4 (Apple M1 Pro, 16 GB RAM). With a time
step $\Delta_t = 0.01$, the discretized dynamics are:
\begin{subequations}\label{eq:sys1}
  \begin{align}
    x_{1}^+ &= x_1 + \Delta_t x_2, \\
    x_2^+ &= x_2 + \Delta_t (x_1 - 3 x_1^3).
  \end{align}
\end{subequations}

%To benchmark numerical accuracy, we simulate $500$ time steps from
%$10$ uniformly sampled initial conditions in $[-2,2]^2$, yielding
%$N = 5000$ points. 
To benchmark numerical accuracy, we generate the dataset $X \in \real^{2 \times N}$ by uniformly sampling $N = 5000$ points in $[-2,2]^2$. This dataset size is deliberately chosen because it remains computationally tractable for exact principal argument computations (Section \ref{sec:exact_RKHS_PA}). For the RKHS, we select a compactly supported Wendland kernel \cite{HW:04} with a smoothness parameter of $\beta = 2$. We then construct a base dictionary $\Vc$ of dimension $s=200$, comprising empirical kernel sections whose centers are chosen by randomly sampling points from the dataset $X$. 

To evaluate the scalability and accuracy of our Nystr\"om approximation, we vary the number of landmark samples, $D \in \{800, 1000, 2000, 3000, 4000\}$. The samples are uniformly drawn from the dataset $X$. For each $D$, we compute the orthonormality residuals:
\begin{align*} 
  \epsilon_{\mathcal{V}} &= \| \tilde{Q}_{\mathcal{V}}^{\top} \tilde{Q}_{\mathcal{V}} - I \|_2 = \|(\tilde{R}_{\mathcal{V}}^\dagger)^{\top} M_{\mathcal{V}} \tilde{R}_{\mathcal{V}}^\dagger - I\|_2, \\ 
  \epsilon_{\mathcal{K} \mathcal{V}} &= \| \tilde{Q}_{\mathcal{K} \mathcal{V}}^{\top} \tilde{Q}_{\mathcal{K} \mathcal{V}} - I \|_2 = \|(\tilde{R}_{\mathcal{K} \mathcal{V}}^\dagger)^{\top} M_{\mathcal{K} \mathcal{V}} \tilde{R}_{\mathcal{K} \mathcal{V}}^\dagger - I\|_2, 
\end{align*}
where $M_{\mathcal{V}}$ and $M_{\mathcal{K} \mathcal{V}}$ are the exact Gram matrices. By quantifying how far the approximate bases deviate from true orthonormality under the RKHS inner product, these residuals directly measure the approximation error in the principal argument computations. The results are summarized in the left plot of Figure \ref{fig:principal_angles}. We see thati ncreasing the number of Nystr\"om samples $D$ steadily decreases the orthonormality residuals $\epsilon_{\mathcal{V}}$ and $\epsilon_{\mathcal{K} \mathcal{V}}$, confirming an improvement in approximation quality. 

Next, we apply both the exact Kernel-SPV and the Approximate Kernel-SPV methods to prune the target subspace $\Sc = \text{span}(\Vc)$. To establish a rigorous benchmark for our scalable approach, we compute the true principal angles of the subspaces identified by the approximate method using the exact RKHS routine. The results, summarized in the right plot of Figure \ref{fig:principal_angles}, demonstrate that with $D = 2000$ Nystr\"om samples, the output of the Approximate Kernel-SPV closely matches the exact baseline. Notably, even though the orthonormality residuals $\epsilon_{\mathcal{V}}$ and $\epsilon_{\mathcal{K} \mathcal{V}}$ remain relatively high at this sample size, the method still successfully extracts the same approximately invariant subspace. This confirms that the Nystr\"om-based approximation efficiently captures the essential subspace geometry required for pruning, well before the empirical bases achieve perfect numerical orthonormality.
%
%\marginJC{It wouldn't hurt (here, or probably earlier in this %section)
%  to remind the briefly remind the reader that by exact %Kernel-SPV we
%  mean the pruning procedure described in Section III together %with
%  Algorithm 1 for principal arguments computation. And %approximate the
%  same procedure but with Algorithm 2.}
%\marginDS{Pending}
%

\begin{figure}[htb!]
  \centering
  \begin{subfigure}[b]{0.49\linewidth}
    \centering
    \includegraphics[width=\linewidth]{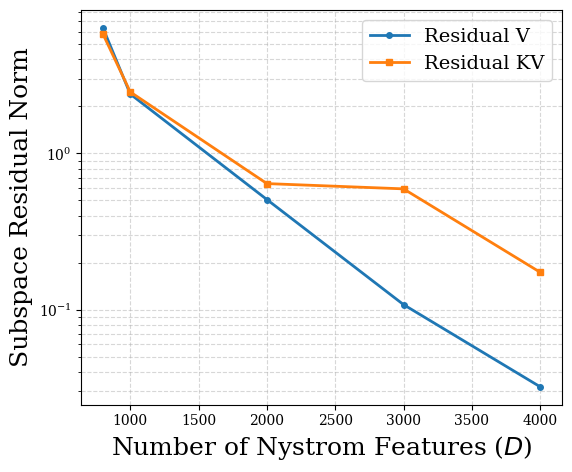}
    \label{fig:exact_angles}
  \end{subfigure}
  \hfill
  \begin{subfigure}[b]{0.49\linewidth}
    \centering
    \includegraphics[width=\linewidth]{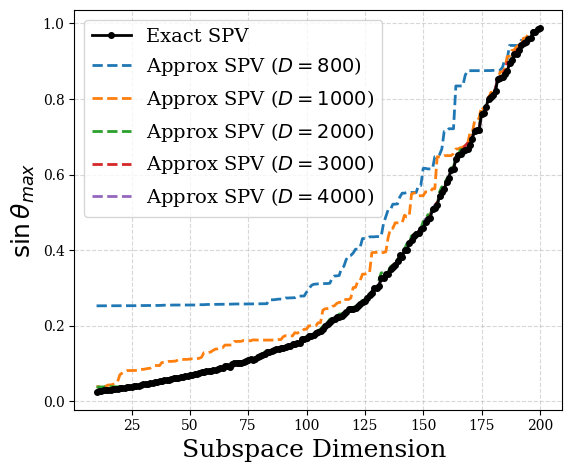}
    \label{fig:approx_angles}
  \end{subfigure}
    \vspace*{-5ex}
  \caption{(Left) Residuals of the approximate orthonormal bases for $\mathcal{V}$ and $\mathcal{K}\mathcal{V}$ as a function of the number of Nystr\"om samples $D$. (Right) Principal angles computed by the exact Kernel-SPV method and the approximate method.}
  \label{fig:principal_angles}
\end{figure}
%
%\marginJC{In the figure on the right, I'd eliminate the dots in the
%  black curve (right now, it short of obscures seeing the other plots well). }
%\marginDS{Resolved. Also changed some parameters to see more %interesting behavior.}
%

In order to demonstrate the practical utility of our pruning strategy, we plot the $5$-step prediction error of the leading Koopman eigenfunction before and after pruning. This corresponds to the eigenfunction with eigenvalue $\lambda \approx 1$, which is the most critical for long-term predictions. The relative prediction error is computed as $| \Kc^5 \phi - \lambda^5 \phi |$ for each point in the dataset, where $\phi$ is the estimated eigenfunction, normalized so that the maximum absolute value of the eigenfunction on the dataset is $1$. 
We utilize $D = 2000$ Nystr\"om samples for the Approximate Kernel-SPV pruning procedure to obtain the pruned subspace of dimension $s^* = 5$. The results are summarized in Figure \ref{fig:lifted_state_predictions}.

\begin{figure}[htb!]
  \centering
  \begin{subfigure}[b]{0.49\linewidth}
    \centering
    \includegraphics[width=\linewidth]{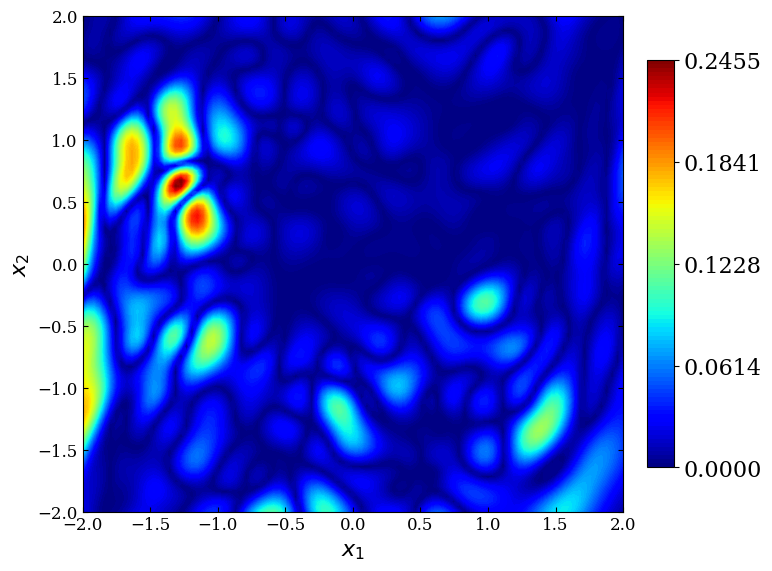}
    \label{fig:unpruned_error}
  \end{subfigure}
  \hfill
  \begin{subfigure}[b]{0.49\linewidth}
    \centering
    \includegraphics[width=\linewidth]{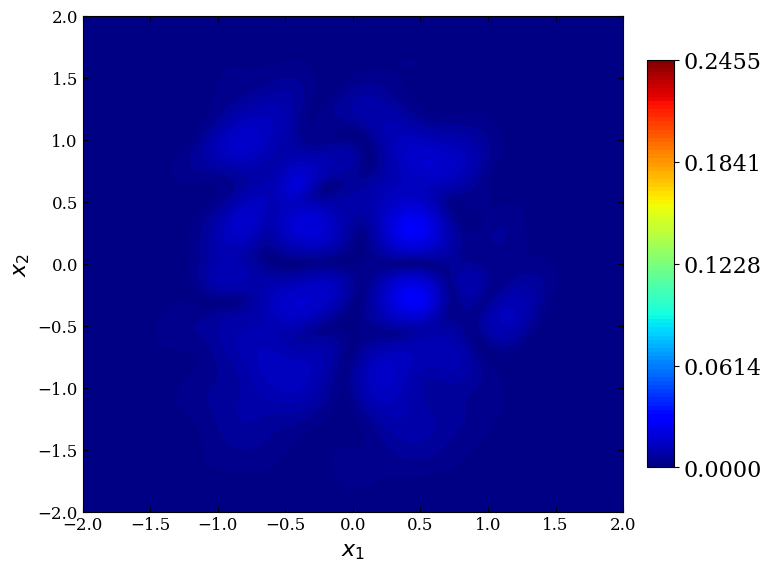}
    \label{fig:pruned_error}
  \end{subfigure}
      \vspace*{-5ex}
  \caption{Relative prediction error $| \Kc^5 \phi - \lambda^5 \phi | $ for the
    estimated eigenfunction $\phi$ corresponding to
    $\lambda \approx 1$, obtained from Kernel EDMD (left) and after
    pruning with Approximate Kernel-SPV (right).}
  \label{fig:lifted_state_predictions}
\end{figure}
%
%\marginJC{In caption, relative prediction error?}
%\marginDS{Resolved.}
%

\vspace*{-2ex}
\section{Conclusions}
We have introduced a rigorous, data-driven framework for Koopman
subspace pruning within an RKHS. We derived an exact computational
routine to calculate principal angles and vectors. To resolve the
computational bottleneck associated with massive trajectory datasets,
we integrated a scalable Nystr\"om approximation. Together, these
methods establish the Kernel-SPV and Approximate Kernel-SPV
algorithms, offering robust solutions for targeted subspace
refinement. Simulations on the undamped Duffing oscillator validate the
numerical accuracy of our Nystr\"om-based approach and demonstrate
that this pruning strategy systematically improves observable
predictions and eigenfunction approximations. Future work will focus on deriving prediction error
bounds based on the computed principal angles and developing more
sophisticated pruning strategies.

\bibliographystyle{ieeetr}
\bibliography{../bib/alias,../bib/Main-add,../bib/Main,../bib/JC}

\end{document}